\documentclass[runningheads]{llncs}
\usepackage{graphicx}
\usepackage{listings}
\usepackage{fancyhdr,graphicx,amsmath,amssymb}
\usepackage[ruled,vlined, linesnumbered]{algorithm2e}

\begin{document}
\renewcommand{\lstlistingname}{Algorithm}
\title{Improved Combinatorial Approximation Algorithms for MAX CUT in Sparse Graphs}
\author{Eiichiro Sato}
\authorrunning{E. Sato}
\institute{Department of Computer Science, The University of Tokyo, Japan
\email{gasin@is.s.u-tokyo.ac.jp}}
\maketitle              % typeset the header of the contribution
\begin{abstract}
The Max-Cut problem is a fundamental NP-hard problem, which is attracting attention in the field of quantum computation these days.
Regarding the approximation algorithm of the Max-Cut problem, algorithms based on semidefinite programming have achieved much better approximation ratios than combinatorial algorithms. Therefore, filling the gap is an interesting topic as combinatorial algorithms also have some merits. 
In sparse graphs, there is a linear-time combinatorial algorithm with the approximation ratio $\frac{1}{2}+\frac{n-1}{4m}$ [Ngoc and Tuza, Comb. Probab. Comput. 1993], which is known as the Edwards-Erd\H{o}s bound. In subcubic graphs, the combinatorial algorithm by Bazgan and Tuza [Discrete Math. 2008] has the best approximation ratio $\frac{5}{6}$ that runs in $O(n^2)$ time.
Based on the approach by Bazgan and Tuza, we introduce a new vertex decomposition of graphs, which we call tree-bipartite decomposition. With the decomposition, we present a linear-time $(\frac{1}{2}+\frac{n-1}{2m})$-approximation algorithm for the Max-Cut problem. As a derivative, we also present a linear-time $\frac{5}{6}$-approximation algorithm in subcubic graphs, which solves an open problem in their paper.
\end{abstract}
\section{Introduction}
The Max-Cut problem, the problem to find the maximum bipartite subgraph of a given graph, is a classical NP-hard problem.
%More precisely, if a graph $G(V,E)$ is given, the problem is to find a partition $(A,B)$ of $V$ where the number of edges between $A$ and $B$ is maximized.
The Max-Cut problem can be defined in weighted graphs, but only unweighted graphs are treated in this paper, which is also called the simple Max-Cut problem. Only connected graphs are considered as an input because extension for disconnected graphs is trivial.
%The Max-Cut problem is theoretically interesting not only for its NP-hardness but also for the submodularity of the cut function or the relation to the Ising model.

Besides its theoretical interest, some practical applications are known for the Max-Cut problem such as circuit design and statistical physics~\cite{Application}.
Approximation algorithms are a common approach to solve NP-hard problems. Although there are known limits on the approximation ratio for the Max-Cut problem~\cite{General Inapprox,UGC Inapprox,3-regular Inapprox}, various approximation algorithms have been developed for many classes of the problem. 
Recently, in the field of Quantum Approximation Optimization Algorithms (QAOA), approximation algorithms for bounded degree graphs have attracted much attention~\cite{QAOA,QAOA-nature}. Along with this, the estimation of the size of the maximum cut and approximation algorithms in the classical computation are also attracting attention. 
For example, the comparison of results from classical computation with those from QAOA in regular graphs with high girths is well summarized in \cite{Quantum&Classical} or in \cite{Basso}.

%It is shown that $\frac{16}{17}$-approximation algorithm does not exists if $P\neq NP$~\cite{General Inapprox} and $0.87856$-approximation is the limit under Unique Game Conjecture~\cite{UGC Inapprox}. Even for subcubic graphs, graphs with maximum degree three, the problem is still NP-hard~\cite{Subcubic NP-hard}. $0.997$ inapproximability is also known for 3-regular graphs if $P\neq NP$~\cite{3-regular Inapprox}.

In this paper, we deal with combinatorial algorithms for the Max-Cut problem. A combinatorial algorithm here means an algorithm that does not include numerical operations such as matrix operations and eigenvalue calculations as subroutines. The advantages of combinatorial algorithms are that they do not require consideration of numerical errors, they are fast when implemented, and they are useful for understanding the structure of the problem~\cite{comb-general}. However, the approximation ratios by the combinatorial algorithms are considerably worse than those based on semidefinite programming (SDP). For general graphs, Kale and Seshadohri~\cite{comb-general} have developed a combinatorial algorithm whose approximation ratio is better than 0.5, but the approximation ratio is worse than $0.87856$-approximation algorithm based on SDP by Goemans and Williamson~\cite{SDP-general}. For subcubic graphs, graphs with the maximum degree three, the SDP-based algorithm achieves a $0.9326$-approximation ratio~\cite{cubic}, while the best approximation ratio by a combinatorial algorithm is $\frac{5}{6}=0.8333$~\cite{BZ}. Since the algorithm by Goemans and Williamson, which is implementable in nearly linear time~\cite{linear-SDP,Spectral}, achieves the optimal approximation ratio under Unique Game Conjecture~\cite{UGC Inapprox}, it is unlikely to exceed the approximation ratios of SDP-based algorithms by combinatorial algorithms, but it is interesting to fill the gaps as the combinatorial algorithms have their own advantages.

%The most successful approach to approximate the maximum cut is based on semidefinite programming (SDP). Goemans and Williamson achieve a $0.87856$ approximation ratio by SDP~\cite{SDP-general}, which is optimal under UGC~\cite{UGC Inapprox}, and it is shown that the algorithm can be implemented in linear time~\cite{linear-SDP,Spectral}. For subcubic graphs, the approach based on SDP achieves the approximation ratio of $0.9326$~\cite{cubic}.
%As another approach, algorithms based on the spectral graph theory are also known~\cite{Spectral,Better Spectral}.

The difficulty of combinatorial approximation algorithms is that they need to find a large size cut and at the same time evaluate the upper bound of the maximum cut. Therefore, most of the previous results by combinatorial algorithms are just to find a large size cut. For general graphs, Ngoc and Tuza~\cite{linear-approx} show a cut with the size $\frac{m}{2}+\frac{n-1}{4m}$ can be constructed in linear time, which is known as the Edwards-Erd\H{o}s bound~\cite{Edwards,Erdos}, and if the maximum degree of vertices is $\Delta$, Hofmeister and Lefmann~\cite{Bounded Comb} show a cut with the size $\frac{m}{2}+\frac{m}{2\Delta}$ can be constructed in linear time. For triangle-free subcubic graphs, Bondy and Locke~\cite{Comb subcubic-triangle-free} show that a cut of the size $\frac{4m}{5}$ can be constructed in $O(n^2)$ time, and Halperin, et al.~\cite{cubic} extend the result by showing that the triangle-free constraint can be removed while preserving the $\frac{4}{5}$-approximation ratio, which utilizes the fact that the maximum cut can not include all edges of a triangle. In subcubic graphs, Bazgan and Tuza~\cite{BZ} construct a $\frac{5}{6}$-approximation algorithm which runs in $O(n^2)$ time by decomposing a subcubic graph into vertex-disjoint unicyclic graphs, graphs with exactly one cycle, and a tree under some conditions. In the analysis of the $\frac{5}{6}$-approximation ratio, the number of edge-disjoint odd cycles are used to bound the size of the maximum cut.

The algorithm by Bazgan and Tuza~\cite{BZ} is to find the maximum cut in each unicyclic graph and include more than half of the edges between unicyclic graphs. In subcubic graphs, if two edges are included in a cut for each vertex and the remaining edges are included in a cut with the probability $\frac{1}{2}$, the approximation ratio $\frac{5}{6}$ can be achieved. As the average degree of a unicyclic graph is two, it seems to achieve the goal. But, if a unicyclic graph contains an odd cycle, the approximation ratio is not simply achieved since the size of the maximum cut is less than the number of edges of the unicyclic graph. Bazgan and Tuza solve this problem by using the fact that the presence of odd cycles suppresses the upper bound on the maximum cut, and that the subcubic property can be used to increase the number of edges between unicyclic graphs which are included in a cut. However, constraints are imposed on the decomposition to achieve the latter goal, which causes the construction of the decomposition to take $O(n^2)$ time. In addition, in unicyclic decomposition, the last component can be a tree, which worsens the approximation ratio. To solve this problem, Bazgan and Tuza use an $O(n^2)$-time recursive algorithm that contains a large constant in the time complexity, which may cause a problem in practice.

Our algorithms are based on the algorithm by Bazgan and Tuza, but we use a new vertex decomposition of a graph instead of unicyclic decomposition. In our decomposition, a tree inducing an odd cycle, which is defined in Section3, is used instead of a unicyclic graph with an odd cycle, and we can achieve the target approximation ratio by simply suppressing the upper bound on the maximum cut. Since our decomposition does not require subcubic property, it is applicable to general graphs and can be constructed in linear time. Although the problem that the last component of the decomposition can be a tree still remains, we solve it by considering it as a special case of the edge bipartization problem~\cite{bipartization}. For general graphs, such an edge bipartization problem takes $O(nm)$ time to solve, but for sufficiently sparse graphs, the problem can be solved in linear time.

\subsection{Our Contribution}
We developed the following combinatorial algorithms for the Max-Cut problem in general graphs, which works well on sparse graphs.
\begin{itemize}
\item{}$O(m)$-time $(\frac{1}{2}+\frac{n-1}{2m})$-approximation algorithm
\item{}$O(nm)$-time $(\frac{1}{2}+\frac{n}{2m})$-approximation algorithm
\end{itemize}
Though the purpose of the previous researches~\cite{Edwards,Bounded Comb} is to bound the size of the maximum cut, as far as we know, there exist no combinatorial algorithms which achieve a better approximation ratio than those. From the perspective of the approximation ratio, the Edwards-Erd\H{o}s bound~\cite{Edwards} means $(\frac{1}{2}+\frac{n-1}{4m})$-approximation, and the bound by Hofmeister and Lefmann~\cite{Bounded Comb} means $(\frac{1}{2}+\frac{1}{2\Delta})$-approximation, which is not better than $(\frac{1}{2}+\frac{n}{4m})$-approximation since $\Delta \geq \frac{2m}{n}$. Therefore, our $(\frac{1}{2}+\frac{n-1}{2m})$-approximation algorithm achieves the best approximation ratio among them without worsening the time complexity.

We also developed a combinatorial algorithm for the Max-Cut problem in graphs with $\frac{m}{n}\leq 2$, which are very sparse.
\begin{itemize}
\item{} $O(n)$-time $(\frac{1}{2}+\frac{n}{2m})$-approximation algorithm
\end{itemize}
As subcubic graphs are graphs where $\frac{m}{n} \leq \frac{3}{2}$, this algorithm runs in $O(n)$ time with the approximation ratio $\frac{5}{6}$ in subcubic graphs, which solves the open problem posed in~\cite{BZ}.

%\subsection{Our Approach}
%Our approach is based on the algorithm by Bazgan and Tuza~\cite{BZ}.

%We use a new graph decomposition inspired by the unicyclic decomposition used in~\cite{BZ}. To treat unicyclic graphs with an odd cycle, their decomposition requires some conditions which cause the efficient construction hard and their decomposition also requires that the maximum degree is at most three. Our idea is to use a tree that induces an odd cycle instead of a unicyclic graph with an odd cycle. Actually, by using the graph as a component of the decomposition, the construction becomes very simple and some restrictions disappear.

 %In our decomposition, the last component can be a tree, which worsens the approximation ratio. Bazgan and Tuza solve the problem by the recursive algorithm which utilizes the property of subcubic graphs. Our approach is to treat the problem as the special case of the edge bipartization problem. By utilizing edge-disjoint odd cycles, we develop efficient algorithms.

\section{Notations}
Let $G(V,E)$ be a connected graph where $V$ is a set of vertices and $E$ is a set of edges. Define $V(G) := V, \ E(G) := E$ and define $G[V]$ as an induced subgraph in $G$ by $V$. Let $H_1,H_2$ be subgraphs of $G$. Define
$$E_G(H_1,H_2) := \{\{u,v\} \in E(G) \ | \ u\in V(H_1) \land v\in V(H_2)\}$$
and for subsets $V_1,V_2\subset V$, define
$$E_G(V_1,V_2) := \{\{u,v\} \in E(G) \ | \ u\in V_1 \land v\in V_2\}.$$
Let $H_1,\dots,H_t$ be a vertex-disjoint decomposition of $G$.
%For $i,j \in [1,t]$, let $G[H_i\cup H_j] := G[V(H_i)\cup V(H_j)]$.
Then, define 
$$G[H_{\geq i}] := G[V(H_i)\cup V(H_{i+1})\cup \dots \cup V(H_t)]$$ and  $$G[H_{>i}] := G[V(H_{i+1})\cup \dots \cup V(H_t)].$$
Let $A,B$ be a partition of $V$. Define $C(A,B)$ as a cut of $G$, and the size of a cut $C(A,B)$ is defined as $|E_G(A,B)|$. 
Define $mc(G)$ as the size of the maximum cut of $G$.

\section{Tree-Bipartite Decomposition}
Tree-bipartite decomposition is a vertex-disjoint decomposition of a graph, which is inspired by unicyclic decomposition~\cite{BZ}. All of our algorithms are based on this decomposition.
The components of unicyclic decomposition are unicyclic graphs and a tree. In tree-bipartite decomposition, trees inducing an odd cycle (IOC trees) and cycle bipartite graphs (CB graphs) are used instead of unicyclic graphs.

\begin{figure}
\includegraphics[width=\textwidth]{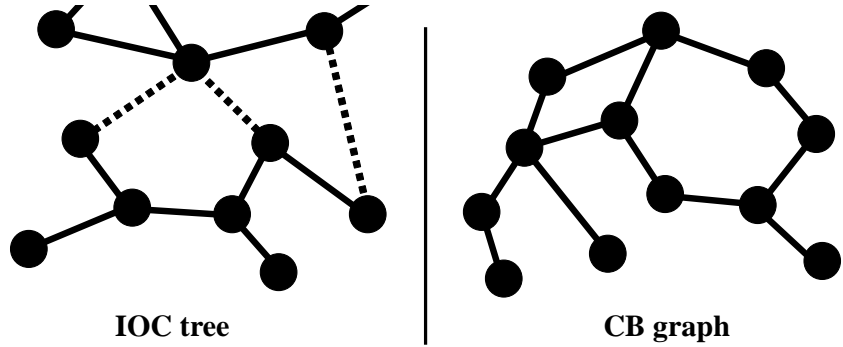}
\caption{Components of tree-bipartite decomposition} \label{fig1}
\end{figure}

A tree inducing an odd cycle can be seen as an extension of a unicyclic graph with an odd cycle. The good property of the tree is that the size of the maximum cut is equal to the number of edges while the tree detects an edge-disjoint odd cycle.
\begin{definition}
For some natural number $t$, let $H_1,\dots,H_t$ be a vertex-disjoint decomposition of a graph $G$. $H_i \ (1\leq i < t)$ is called a tree inducing an odd cycle if $H_i$ is a tree and there exists $r\in V(G[H_{>i}])$ such that $G[V(H_i)\cup\{r\}]$ contains an odd cycle (we call such a vertex a root of $H_i$).
We abbreviate a tree inducing an odd cycle with an IOC tree.
\end{definition}

Then, define a cyclic bipartite graph, which can be seen as an extension of a unicyclic graph with an even cycle.

\begin{definition}
A cyclic bipartite graph is a connected bipartite graph with at least one cycle. We abbreviate a cyclic bipartite graph with a CB graph. 
\end{definition}

Then, define the tree-bipartite decomposition. Note that the last component can be a tree. The other components are an IOC tree or a CB graph.

\begin{definition}
Given a graph $G$, let $t$ be a natural number and $H_1,\dots,H_t$ be a vertex-disjoint decomposition of $G$. $H_1,\dots,H_t$ is called a tree-bipartite decomposition if the following conditions are satisfied.\\
- for any $1\leq i < t$, $H_i$ is an IOC tree or a CB graph \\
- $H_t$ is a CB graph or a tree \\
- for any $1\leq i < t$, there exists $j > i$ such that $|E_G(H_i,H_j)| \geq 1$
\end{definition}

\begin{figure}
\includegraphics[width=\textwidth]{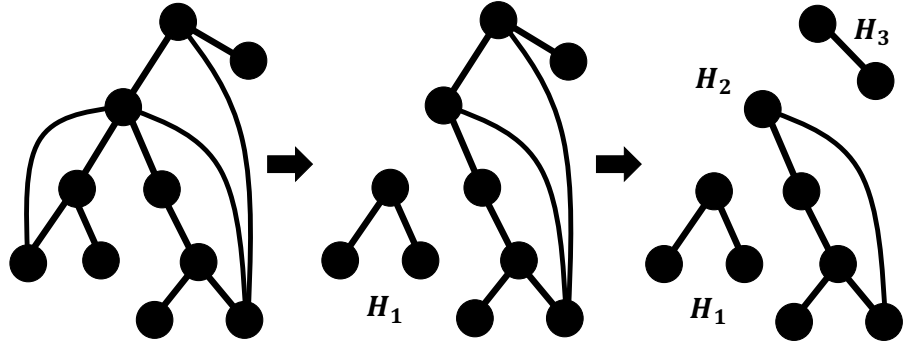}
\caption{Example of tree-bipartite decomposition. $H_1$ is an IOC tree, $H_2$ is a CB graph, and $H_3$ is a tree.} \label{fig3}
\end{figure}

Compared to the unicyclic decomposition~\cite{BZ}, the third condition is relaxed. Thereby, the construction of a tree-bipartite decomposition is simply implemented in linear time.

\begin{lemma}\label{lem1}
A tree-bipartite decomposition of a given graph $G(V,E)$ can be constructed in linear time.
\end{lemma}
\begin{proof}
Let $H = \emptyset$.
Construct a depth-first search tree $T$ from an arbitrary vertex. 
Then, iterate the vertices in $V$ in descending order of the preorder.
Let $r$ be the vertex that is being looked at. Let $S_1, \dots, S_d$ be the subtrees of the children of $r$ in $T$. For $1\leq i \leq d$, if $G[V(S_i)\cup\{r\}]$ contains an odd cycle, append $S_i$ at the end of $H\dots (i)$ and update $T \leftarrow T[V(T)\setminus V(S_i)]$. Let $S$ be the subtree of $r$ in $T$ after the operation against each subtree. If $G[V(S)]$ contains a cycle, append $G[V(S)]$ at the end of $H\dots (ii)$ and update $T \leftarrow T[V(T)\setminus V(S)]$. After the iteration of the vertices, if $V(T) \neq \emptyset$, append $T$ at the end of $H\dots (iii)$.

Both construction of a depth-first search tree and detection of an odd cycle can be done in $O(m)$ time. Therefore, the algorithm runs in $O(m)$ time as a whole.

Then, let us check the correctness.
The important property is that the vertices are seen in descending order of the preorder and that all cycles which appeared are removed from $T$. Therefore, the graph $G[V(S)\setminus \{r\}]$ does not contain cycles where $r$ is the vertex that is being looked at and $S$ is the subtree of $r$.
Then, graphs appended at $(i)$ are IOC trees, and a graph appended at $(iii)$ is a tree, and as all odd cycles are removed at $(i)$, graphs appended at $(ii)$ are CB graphs.
For the third condition, if $H_i$ is an IOC tree, it is obvious from the definition. If $H_i$ is a CB graph and it does not contain the root of the depth-first search tree, it cuts the edge in the tree. If $H_i$ is a CB graph and it contains the root of the tree, then $i=t$, which contradicts the condition $1\leq i < t$.
\qed
\end{proof}

\section{Approximation Algorithms}
The construction of the cut and its analysis are very similar to~\cite{BZ}.
The algorithm in~\cite{BZ} does additional operation when dealing with unicyclic graphs with an odd cycle, since it is required to increase the number of edges between unicyclic graphs which are included in a cut.
Since we do not need to do such operations when merging IOC trees or CB graphs, our construction is simpler.

\begin{theorem}\label{thm1}
There is an $O(m)$-time combinatorial $(\frac{1}{2}+\frac{n-1}{2m})$-approximation algorithm for the Max-Cut problem.
\end{theorem}
\begin{proof}
Given a graph $G(V,E)$, let $H_1, \dots, H_t$ be a tree-bipartite decomposition of $G$ constructed by Lemma~\ref{lem1}. Iterate $H_1,\dots,H_t$ in reverse order and construct a cut $C(A,B)$ which is initialized as $C(\emptyset,\emptyset)$. Assuming that $C$ is constructed for the graph $G[H_{>i}]$, construct $C$ for the graph $G[H_{\geq i}]$. Let $C_i(A_i,B_i)$ be one of the maximum cuts of $H_i$. If $E_G(A_i,A)+E_G(B_i,B) \geq E_G(A_i,B)+E_G(B_i,A)$, update $C(A,B)$ to $C(A\cup B_i, B\cup A_i)$. Otherwise, update $C(A,B)$ to $C(A\cup A_i, B\cup B_i)$.

Computing the maximum cut of an IOC tree, a CB graph, or a tree can be done in linear time, and the tree-bipartite decomposition is constructed in linear time by Lemma~\ref{lem1}. Therefore, the time complexity of this algorithm is $O(m)$.

Then, let us analyze the approximation ratio. Suppose $H_t$ is a tree. For $1\leq i \leq t$, if $H_i$ is a tree or an IOC tree, $|E(H_i)| = |V(H_i)|-1$ and $mc(H_i) = |V(H_i)|-1$.
If $H_i$ is an CB graph, $|E(H_i)| \geq |V(H_i)|$ and $mc(H_i) = |E(H_i)|$.
Let $$c = \sum_{H_i \text{ is an CB graph}}|E(H_i)|-|V(H_i)|$$ and let $E_{out}$ be a set of edges which are not included in $E(H_j)$ for any $1\leq j \leq t$.
By its construction, $C$ contains at least a half of edges in $E_{out}$.
Therefore, if $x$ represents the number of IOC trees in $H$, the cut size of $C$ is at least
\begin{align*}
n+c-x-1 + \frac{1}{2}|E_{out}|&=n+c-x-1+\frac{1}{2}(m-n-c+x+1) \\
&= \frac{m+n+c-x-1}{2} \\
&\geq \frac{m+n-x-1}{2} \ (\because c \geq 0).
\end{align*}
When $H_i$ is an IOC tree for some $1\leq i \leq t$, $G[H_{\geq i}]$ contains at least one odd cycle which is edge-disjoint with $G[H_{> i}]$ by the definition of an IOC tree.
Therefore, $$mc(G) \leq m-x.$$ Then, the approximation ratio is
\begin{align*}
(\text{approximation ratio}) &\geq \frac{m+n-x-1}{2(m-x)} \\
&= \frac{1}{2}+\frac{n-1}{2(m-x)} \\
&\geq \frac{1}{2} + \frac{n-1}{2m}.
\end{align*}
If $H_t$ is a CB graph, the cut size becomes $\frac{m+n-x}{2}$ and the approximation ratio is better than the case when $H_t$ is a tree.
\qed
\end{proof}

\begin{theorem}\label{thm2}
There is an $O(nm)$-time combinatorial $(\frac{1}{2}+\frac{n}{2m})$-approximation algorithm for the Max-Cut problem.
\end{theorem}

Though the time complexity becomes worse, Theorem~\ref{thm2} improves the approximation ratio of Theorem~\ref{thm1}. The idea is to remove a tree from a tree-bipartite decomposition by merging it with another component to make the generated graph contain an even cycle, which guarantees a decent size of the maximum cut. This idea is also used in~\cite{BZ}, but a different approach is used.

Before showing Theorem~\ref{thm2}, let us show some lemmas. Lemma~\ref{lem2} and Lemma~\ref{lem3} describe the condition to slightly improve the approximation ratio of Theorem~\ref{thm2}. 

\begin{lemma}\label{lem2}
Given a graph $G(V,E)$, let $H_1,\dots,H_t$ be a tree-bipartite decomposition. 
For some $1\leq i \leq t$, let $H' = G[H_{\geq i}]$ where $H'$ contains an even cycle.
Suppose the maximum cut of $H'$ is given.
Then, a $(\frac{1}{2}+\frac{n}{2m})$-approximation algorithm can be constructed in $O(m)$ time.
\end{lemma}
\begin{proof}
Run the same algorithm as Theorem~\ref{thm1} for $H_1,\dots,H_{i-1}$ where the initial cut is the maximum cut of $H'$.
Let $x$ be the number of IOC trees in $H_1,\dots,H_{i-1}$.
Let $m' = |E(H')|, n' = |V(H')|$ and let $mc(H') = m'-l$ where $l\in \mathbb{N}$.
As $H'$ contains an even cycle, $mc(H') = m'-l \geq n'$. 
Then, the approximation ratio can be computed as follows.
\begin{align*}
(\text{approximation ratio}) &\geq \frac{n-x+m'-l-n'+\frac{1}{2}(m-n+x-m'+n')}{m-x-l} \\
&= \frac{m+n-x+m'-n'-2l}{2(m-x-l)} \\
&\geq \frac{(m+n-x-l)}{2(m-x-l)} \ (\because m' \geq n' + l)\\
&= \frac{1}{2} + \frac{n}{2(m-x-l)} \geq \frac{1}{2}+\frac{n}{2m}
\end{align*}
\qed
\end{proof}

\begin{lemma} \label{lem3}
Given a graph $G(V,E)$, let $H_1,\dots,H_t$ be a tree-bipartite decomposition. 
For some $1\leq i \leq t$, let $y$ be the number of IOC trees in $H_i,\dots,H_t$.
If $mc(G[H_{\geq i}]) < |E(G[H_{\geq i}])|-y$, a $(\frac{1}{2}+\frac{n}{2m})$-approximation algorithm can be constructed in $O(m)$ time.
\end{lemma}
\begin{proof}
Run the same algorithm as Theorem~\ref{thm1}. Let $x$ be the number of IOC trees in the tree-bipartite decomposition. From the assumption, $mc(G) \leq m-x-1$. Then,
\begin{align*}
(\text{approximation ratio}) &\geq \frac{m+n-x-1}{2(m-x-1)} \\
&= \frac{1}{2}+\frac{n}{2(m-x-1)} \\
&\geq \frac{1}{2} + \frac{n}{2m}.
\end{align*}
\qed
\end{proof}

Lemma~\ref{lem4} utilizes the property that the structure of a graph without even cycles is similar to a tree, and a technique like dynamic programming is applicable.

\begin{lemma} \label{lem4}
Let $G(V,E)$ be a connected graph without even cycles and $A,B \subset V$ be disjoint sets. Let $y$ be the number of odd cycles in $G$. Then, the existence of the cut $C(A', B')$ with the size $m-y$ where $A\subset A'$ and $B\subset B'$ can be decided in $O(m)$ time. If such a cut exists, the cut can be constructed in $O(m)$ time.
\end{lemma}
\begin{proof}
Let $H_1,\dots,H_t$ be a tree-bipartite decomposition of $G$. As $G$ does not contain even cycles, $H_1,\dots,H_{t-1}$ are IOC trees and $H_t$ is a tree. Moreover, for each $H_i$, there exists only one root of $H_i$ and $|E_G(H_i,G[H_{>i}])| = 2$ because there must be an even cycle in $G$ if there exist more than two roots or $|E_G(H_i,G[H_{>i}])|>2$.

Iterate $H$ from $H_1$ to $H_{t-1}$. For $1\leq i \leq t-1$, let $r_i$ be the root of $H_i$. Since $H_i$ is a (IOC) tree, we can judge whether there exists a cut $C_i(A_i,B_i)$ of $H_i$ with the size $|E(H_i)|-1$ or not where $A_t\cap B = \emptyset$ and $B_t\cap A = \emptyset$. Let us consider both cases $r_i \in A_i$ and $r_i \in B_i$. If $C_i$ does not exist in both cases, it is shown that there does not exist a cut of $G$ with the size $m-y$. If $C_i$ exists in both cases, leave $r_i$ unfixed. If $C_i$ exists in only one case, append $r_i$ to $A$ if $r_i \in A_i$ and append $r_i$ to $B$ otherwise. If the iteration ends, judge whether there exists a cut $C_t(A_t,B_t)$ of $H_t$ with the size $|E(H_t)|-1$ or not where $A_t\cap B = \emptyset$ and $B_t\cap A = \emptyset$. If such a cut exists, the cut of $G$ with the size $m-y$ can be constructed from $H_t$ to $H_1$. If such a cut does not exist, it is shown that there does not exist a cut of $G$ with cut size $m-y$.

Even if the assignment of some vertices are fixed, the maximum cut of an IOC tree can be computed in linear time by greedily constructing a cut from leaves. Therefore, the time complexity is $O(m)$ in total.
\qed
\end{proof}

\begin{figure}
\includegraphics[width=\textwidth]{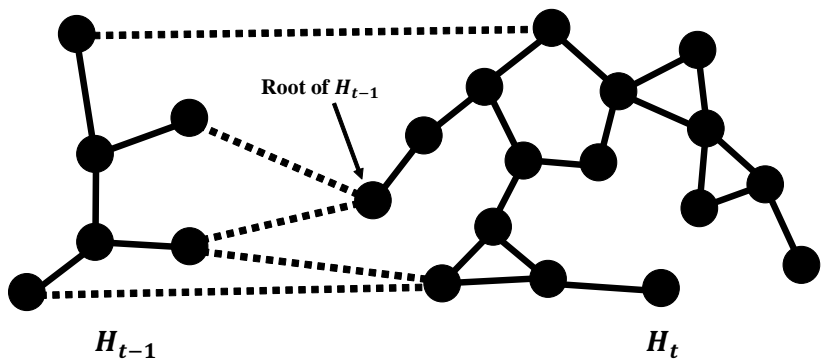}
\caption{Merging an IOC tree and a graph without even cycles} \label{fig2}
\end{figure}

Then, let us prove Theorem~\ref{thm2} with the previous lemmas.

\begin{proof}
Let $H_1,\dots,H_t$ be a tree-bipartite decomposition of $G$.
As long as $G[V(H_t)\cup V(H_{t-1})]$ does not contain an even cycle, update $H_{t-1} \leftarrow G[V(H_t)\cup V(H_{t-1})]$ and $t \leftarrow t-1$. Then, $H_t$ is a tree or a graph with edge-disjoint odd cycles. Let $y$ be the number of odd cycles in $H_t$.

Suppose $t=1$. Since $G=H_t$ does not contain even cycles, $mc(G) = |E|-1$. Therefore, the maximum cut can be constructed in linear time by constructing a spanning tree of $G$.

Then, let us consider the case $t>1$. Let $H' = G[V(H_{t-1})\cup V(H_t)]$.

Suppose $H_{t-1}$ is a CB graph. If the maximum cut of $H'$ is constructed or $mc(H') < |E(H')|-y$ is shown, by Lemma~\ref{lem2} and Lemma~\ref{lem3}, the approximation ratio $\frac{1}{2}+\frac{n}{2m}$ is achieved. Consider the graph $G'$ that is $H_{t-1}$ with $E_G(H_{t-1},H_t)$ and its endpoints. As $G'$ is edge-disjoint with $H_t$, $G'$ should be bipartite to satisfy $mc(H') = |E(H')|-y$. If $G'$ is bipartite, by fixing the assignment of $V(G')$ with the maximum cut of $G'$, Lemma~\ref{lem4} can be applied to $H_t$ to construct the maximum cut of $H'$ or to show $mc(H') < |E(H')|-y$.

Suppose $H_{t-1}$ is an IOC tree (Fig.~\ref{fig2}). If the maximum cut of $H'$ is constructed or $mc(H') < |E(H')|-y-1$ is shown, by Lemma~\ref{lem2} and Lemma~\ref{lem3}, the approximation ratio $\frac{1}{2}+\frac{n}{2m}$ is achieved. Consider the graph $G'$ that is $H_{t-1}$ with $E_G(H_{t-1},H_t)$ and its endpoints. As $G'$ is edge-disjoint with $H_t$, $mc(G')$ should be $|E(G')|-1$ to satisfy $mc(H') = |E(H')|-y-1$. Let $U$ be an odd cycle in $G'$. If there exists a cut of $G'$ with the size $|E(G')|-1$, the only edge which is not in the cut should be in $U$. Therefore, for each edge  $e$ in $U$, check the graph $(V(G'),E(G')\setminus \{e\})$ is bipartite or not, and if it is bipartite, fix the assignment of $V(G')$ with the maximum cut and apply Lemma~\ref{lem4} to $H_t$. As the number of edges in $U$ is $O(n)$, the time complexity is $O(nm)$.
\qed
\end{proof}

For graphs with the maximum degree $\Delta$, Corollary~\ref{col1} immediately follows from Theorem~\ref{thm2} since $\frac{m}{2n} \leq \Delta$.
\begin{corollary}\label{col1}
There is an $O(nm)$-time combinatorial $(\frac{1}{2}+\frac{1}{\Delta})$-approximation algorithm for the Max-Cut problem.
\end{corollary}

In Theorem~\ref{thm2}, the bottleneck is merging an IOC tree. However, if $\frac{m}{n}$ is bounded, we can bound the number of odd cycles in the graph, which results in the speedup of the algorithm. 

\begin{theorem}\label{thm3}
There is an $O(n)$-time combinatorial $(\frac{1}{2}+\frac{n}{2m})$-approximation algorithm for the Max-Cut problem in graphs where $m \leq 2n$.
\end{theorem}
\begin{proof}
Suppose the same algorithm is run as Theorem~\ref{thm1}. Let $x$ be the number of IOC trees in the tree-bipartite decomposition.
\begin{align*}
(\text{approximation ratio}) &\geq  \frac{1}{2}+\frac{n-1}{2(m-x)} \\
&= \frac{1}{2} + \frac{n}{2m} + \frac{-1+\frac{n}{m}x}{2(m-x)} \\
&\geq \frac{1}{2} + \frac{n}{2m} + \frac{-1+\frac{1}{2}x}{2(m-x)}.
\end{align*}
Therefore, if $x \geq 2$, the approximation ratio $\frac{1}{2}+\frac{n}{2m}$ is already achieved.

Let us consider the case $x \leq 1$.
The bottleneck of the algorithm in Theorem~\ref{thm2} is when $H_{t-1}$ is an IOC tree. As $x \leq 1$, $H_t$ is a tree when $H_{t-1}$ is an IOC tree.
Let $H'=G[V(H_{t-1})\cup V(H_t)]$.
If the maximum cut of $H'$ is constructed or $mc(H') < |E(H')|-1$ is shown, by Lemma~\ref{lem2} and Lemma~\ref{lem3}, the approximation ratio $\frac{1}{2}+\frac{n}{2m}$ is achieved. Let $r\in H_t$ be a root of $H_{t-1}$, and let $e_1,e_2\in E_G(V(H_{t-1}),\{r\})$ be edges such that $e_1 \neq e_2$ and the graph $(V(H_{t-1})\cup \{r\}, E(H_{t-1})\cup \{e_1,e_2\})$ contains an odd cycle (the existence of $r$ and $e_1,e_2$ is guaranteed by the definition of an IOC tree). Consider the graph $G'$ that is $H_{t}$ with $E_G(H_{t-1},H_t)\setminus \{e_1,e_2\}$ and its endpoints. As $G'$ is edge-disjoint with $H_{t-1}$, $G'$ should be bipartite to be $mc(H') = |E(H')|-1$. If $G'$ is bipartite, by fixing the assignment of $V(G')$ with the maximum cut of $G'$, Lemma~\ref{lem4} can be applied to the graph $(V(H_{t-1})\cup \{r\}, E(H_{t-1})\cup \{e_1,e_2\})$.
\qed
\end{proof}

From Theorem~\ref{thm3}, Corollary~\ref{col2} and Corollary~\ref{col3} follow immediately.
Especially, Corollary~\ref{col3} solves the open problem in \cite{BZ}, which asks the existence of the linear-time combinatorial $\frac{5}{6}$-approximation algorithm for the Max-Cut problem in subcubic graphs.

\begin{corollary}\label{col2}
There is an $O(n)$-time combinatorial $\frac{3}{4}$-approximation algorithm for the Max-Cut problem in graphs with the maximum degree four.
\end{corollary}

\begin{corollary}\label{col3}
There is an $O(n)$-time combinatorial $\frac{5}{6}$-approximation algorithm for the Max-Cut problem in subcubic graphs.
\end{corollary}

\section{Open Problems}
We present an $O(nm)$-time $(\frac{1}{2}+\frac{n}{2m})$-approximation algorithm.
If there exists a linear-time $(\frac{1}{2}+\frac{n}{2m})$-approximation algorithm, it is theoretically interesting. As regular graphs with high girths are well studied in the field of quantum computation~\cite{Quantum&Classical}, optimizing our algorithm for such graphs is also one of the interesting future directions.

There is still a large gap between the approximation ratio of combinatorial algorithms and that of algorithms based on semidefinite programming. Showing some inapproximability of combinatorial algorithms or filling the gap by developing a better combinatorial approximation algorithm may help the interpretation of the gap.


\end{document}